\documentclass[pamm,a4paper,fleqn]{w-art}
\usepackage{times,cite,w-thm}
\usepackage[T1]{fontenc}
\usepackage[utf8]{inputenc}
\usepackage{amsfonts}
\usepackage{graphicx}

\makeatletter
\def\copyrightinfo{}
\def\@copyrightline{}

\def\ps@titlepage{%
  \let\@oddhead\@empty
  \let\@evenhead\@empty
  \let\@oddfoot\@empty
  \let\@evenfoot\@empty
}
\makeatother
\begin{document}
\def\leftmark{M. Papace et al.: Probing for trace estimation in LQCD}
\def\rightmark{Preprint}

%%
%%    The information for the title page will be placed between
%%    \begin{document} and \maketitle. The order of most entries
%%    is determined by the class file and cannot be changed by
%%    rearranging them. The maketitle command follows after the
%%    abstract.
%%
%%    The following commands will be updated by the publisher:
%%
%%    \renewcommand{\copyrightyear}{2016}
%%    \DOIsuffix{pamm.20161zzzz}
%%    \Volume{16} 
%%    \Year{2016} 
%%    \pagespan{1}{}
%%
%%    The short title is optional:

\TitleLanguage[EN]
\title[Probing and graph coloring for trace estimation in LQCD]
      {Probing and graph coloring techniques for trace estimation in Lattice QCD}

%% Please delete not needed author entries.
%% Information for the first author.
\author{\firstname{Mario} \lastname{Papace}\inst{1,3,} \footnote{\ElectronicMail{papace@uni-wuppertal.de}}} \address[\inst{1}]{\CountryCode[DE]Bergische Universität Wuppertal}

%%
%%    Information for the second author
\author{\firstname{Andreas} \lastname{Frommer}\inst{1,}%
\footnote{\ElectronicMail{frommer@uni-wuppertal.de}}}
%%
%%    Information for the third author
\author{\firstname{Jose} \lastname{Jimenez-Merchan}\inst{1,}%
\footnote{\ElectronicMail{jimenezmerchan@uni-wuppertal.de}}}
\author{\firstname{Bruno} \lastname{Lang}\inst{1,}%
\footnote{\ElectronicMail{lang@uni-wuppertal.de}}}
\author{\firstname{Gustavo} \lastname{Ramirez-Hidalgo}\inst{2,}%
\footnote{\ElectronicMail{g.ramirez.hidalgo@fz-juelich.de}}} 
\address[\inst{2}]{\CountryCode[DE]J{\"u}lich Supercomputing Centre, Wilhelm-Johnen-Straße, J{\"u}lich}
\author{\firstname{Christian} \lastname{Schneider}\inst{1,3,}%
\footnote{\ElectronicMail{schneider.christian@ucy.ac.cy}}} 
\address[\inst{3}]{\CountryCode[CY]University of Cyprus, 1 Panepistimiou Avenue, Aglantzia}
%%
%%    \dedicatory{This is a dedicatory.}
%%
%%    Abstract is required.
\AbstractLanguage[EN]
\begin{abstract}
The computation of $\mathrm{Tr}[D^{-1}]$, where $D$ is the Wilson-Dirac matrix of Lattice QCD, is a fundamental and computationally demanding task with applications to disconnected hadronic correlation functions. Since $D^{-1}$ is a dense matrix of prohibitive size, its trace cannot be computed exactly, and one must resort to stochastic estimation via the Hutchinson estimator. The variance of the resulting estimation, however, can be large, as it is dominated by the off-diagonal entries of $D^{-1}$. We review the stochastic probing technique, which reduces the variance by constructing structured sampling vectors from distance-$d$ colorings of the graph associated with $D$, exploiting the exponential off-diagonal decay of $D^{-1}$ to eliminate dominant short-range contributions to the variance. We then present a novel multiplier-based coloring scheme, which achieves valid distance-$d$ colorings at arbitrary distances with significantly fewer colors than the established hierarchical probing construction. We prove that at any intermediate coloring falling between two consecutive hierarchical levels, the multiplier-based estimator achieves strictly lower variance than the partial hierarchical estimator, for large enough $d$. This is confirmed by numerical experiments showing that the multiplier-based variance decreases smoothly and monotonically with the number of colors, avoiding the irregular behavior affecting hierarchical probing at intermediate colorings, and achieving a substantial improvement in relative accuracy.
\end{abstract}
%% maketitle must follow the abstract.
\maketitle                   % Produces the title.
\thispagestyle{empty}

\section{Introduction} \label{sec:intro}

The computation of the trace of a matrix function $f(A)$, where $A \in \mathbb{C}^{n \times n}$ is a large sparse matrix, is a fundamental problem in numerical linear algebra with applications in many areas of computational science. Functions of interest include the matrix logarithm ~\cite{SaibabaAlexanderianIpsen2017}, arising in the computation of the matrix determinant $\log (\det A) = \mathrm{Tr}[\log A]$, the matrix square root and its inverse ~\cite{Frommer2024}, relevant in the context of lattice field theories, and the matrix exponential ~\cite{AlMohyHigham2011}, which appears in the simulation of quantum many-body systems. In all these cases, the matrix $A$ is too large for $f(A)$ to be formed explicitly, and one must resort to methods that access $f(A)$ only implicitly, through matrix-vector products of the form $f(A)v$ for a given vector $v$.

In this work we focus on the particular case $f(A) = A^{-1}$, and in particular on the computation of $\mathrm{Tr}[D^{-1}]$, where $D$ is the Wilson-Dirac matrix arising in Lattice Quantum Chromodynamics (LQCD). This quantity and its variants appear as part of the computation of fundamental observables in a wide range of LQCD applications. The disconnected contribution to hadronic correlation functions~\cite{Bali2010}, which governs the flavor-singlet sector of the hadron spectrum, requires the evaluation of $\mathrm{Tr}[\Gamma D^{-1}]$ for various choices of the matrix $\Gamma$. The scalar condensate $\langle \bar{\psi}\psi \rangle = -\mathrm{Tr}[D^{-1}]/V$, where $V$ is the spacetime volume, is an order parameter for chiral symmetry breaking~\cite{Banks1980} and is directly proportional to $\mathrm{Tr}[D^{-1}]$. The quark number susceptibility, relevant to the study of the QCD phase diagram and the deconfinement transition~\cite{Aoki2006}, likewise reduces to traces of $D^{-1}$ and its derivatives with respect to the chemical potential.

The Wilson-Dirac matrix $D \in \mathbb{C}^{n \times n}$, with $n = 12 \cdot T \cdot  L^3$, is defined on a four-dimensional Euclidean spacetime toroidal lattice of temporal extent $T$ and spatial extent $L$, and has a characteristic block structure that reflects the internal degrees of freedom of the quark field. Each lattice site $x$ carries $N_s  N_c = 12$ degrees of freedom, corresponding to $N_s = 4$ Dirac spin components and $N_c = 3$ color charges, so the quark field at site $x$ is a vector $\psi(x) \in \mathbb{C}^{12}$. Accordingly, $D$ is composed of $12 \times 12$ blocks $D_{xy} \in \mathbb{C}^{12 \times 12}$, with nonzero blocks only between nearest-neighbor sites on the lattice. This makes $D$ a large sparse matrix, with $\mathcal{O}(1)$ nonzero blocks per row, but its inverse $D^{-1}$ is a dense matrix. For realistic lattice volumes currently used in simulations, $n$ can reach $\mathcal{O}(10^9)$ or larger.
%, making any direct computation of $D^{-1}$, which would require $\mathcal{O}(n^3)$ arithmetic operations and $\mathcal{O}(n^2)$ memory, completely impractical. Even forming a single diagonal entry of $D^{-1}$ requires one full linear solve with $D$, making exact computation of $\mathrm{Tr}[D^{-1}]$ via direct summation prohibitively expensive. 
This makes an explicit computation of all diagonal entries of $D^{-1}$ practically infeasible: For each such entry we have to solve one system with $D$, and even with an optimal (multigrid) solver (cf.~\cite{DDalphaAMG-paper}, e.g.)  of complexity $\mathcal{O}(n)$ this would result in a prohibitively large total cost of $\mathcal{O}(n^2)$. 
%Extracting the full trace from individual solves would require $n$ of them, each at a cost of $\mathcal{O}(n)$ thanks to the sparsity of $D$ and the availability of optimal multigrid solvers~\cite{DDalphaAMG-paper}, making the total 
%cost of $\mathcal{O}(n^2)$ entirely prohibitive for the lattice sizes encountered in practice.

The standard approach to circumvent this bottleneck is {\em stochastic trace estimation}. The Hutchinson estimator~\cite{Hutchinson1990} expresses the trace as an expectation value,
\begin{equation}\label{eq:hutchinson}
    \mathrm{Tr}[D^{-1}] = \mathbb{E}\!\left[x^\dagger D^{-1} x\right]
    \approx \frac{1}{N} \sum_{j=1}^{N} x_j^\dagger D^{-1}\, x_j,
\end{equation}
where $x_j \sim U(-1,+1)^n$ are Rademacher vectors, i.e.\ independent random vectors with independent and identically
distributed (i.i.d.) entries taking the value $\pm 1$ with equal probability $\frac{1}{2}$.  
The computation of each sample requires one linear system solve, where $D$ is the matrix of the system and $x$ is the 
right-hand side. This can be done through an iterative solver based on Krylov subspace methods or multigrid 
solvers~\cite{DDalphaAMG-paper}. The estimator is unbiased, and its variance is given by ~\cite{Hutchinson1990}
\begin{equation}\label{eq:variance}
    \mathbb{V}\!\left[x^\dagger D^{-1} x\right] = \frac{1}{2}
    \left\|\mathrm{offdiag}(D^{-1} + (D^{-1})^T)\right\|_F^2,
\end{equation}
so the statistical error is governed entirely by the off-diagonal entries of $D^{-1}$. A key structural property that can be exploited to reduce this variance is that the entries of $D^{-1}$ exhibit exponential decay with the graph distance on the lattice~\cite{Benzi1999}: 
\[
|D^{-1}_{ij}| \sim e^{-m_\pi \, d(i,j)},
\]
where $d(i,j)$ is the graph distance between sites $i$ and $j$ on the lattice and $m_{\pi} = am_{\pi}^{\text{phys}}$, where $a$ denotes the lattice spacing and $m_{\pi}^{\text{phys}}$ the physical pion mass ~\cite{Hernandez1999}. Here we work in natural units, in which the speed of light $c$ and the reduced Planck constant $\hbar$ are set to $c = \hbar = 1$, so that mass, energy, and inverse length share the same units and $m_\pi$ is dimensionless. The quantity $m_\pi$ 
sets the effective decay rate of the quark propagator in QCD~\cite{Luscher2010}: a larger pion mass corresponds to faster decay of the off-diagonal entries of $D^{-1}$, so that the variance~\eqref{eq:variance} is dominated by short-range 
off-diagonal contributions, with those at graph distance $d$ being exponentially suppressed.

This observation motivates the technique of \emph{stochastic probing} ~\cite{Frommer-probing}. The idea is to replace the random Rademacher vectors $x^{(i)}$ by structured vectors constructed from a \emph{distance-$d$ coloring} of the lattice: every lattice site is assigned a color in such a way that lattice sites that are connected via $d$ or less edges do not share the same color. Assigning $\pm 1$ values only to the nodes of a given color class, and zero 
to all others, ensures that all contributions to the variance from node pairs at distance at most $d$ cancel exactly in~\eqref{eq:variance}, eliminating the dominant terms and leaving only the exponentially small tail. The effectiveness of probing thus increases with the coloring distance $d$.

A practically important way to obtain a coloring of the lattice is \emph{hierarchical probing}~\cite{Stathopoulos-hp, Stathopoulos-hp2}. It  
organizes the coloring in a sequence of nested levels with distances $d = 2^0, 2^1, 2^2, \ldots$, doubling the coloring distance from one level to the next. The hierarchical structure ensures that solutions computed at outer levels in the 
recursion can be recycled as the distance increases, reducing the incremental cost of moving to a finer coloring. 
In this hierarchical construction, the number of colors for a distance $d=2^{i-1}$ coloring is $n_c = 2^{n_D(i-1)+1}$ for an $n_D$-dimensional toroidal lattice. 
%$at level $i \in \mathbb{N^+}$, which corresponds to $d = 2^{i-1}$ for an $n_D$-dimensional lattice, and 
%
So $n_c$ can quickly become prohibitively large as the coloring distance increases. In practice, limited computational budgets often prevent reaching the next full hierarchical level, forcing the use of an intermediate partial coloring that does not constitute a valid distance-$2^{(i-1)}$ coloring and yields only irregular and non-monotone variance.

In this work we apply a novel recent coloring scheme, called \emph{multiplier-based coloring}, to the estimation of $\mathrm{Tr}[D^{-1}]$, where $D$ is the Wilson--Dirac matrix, demonstrating that it achieves efficient distance-$d$ colorings with significantly fewer colors than the hierarchical construction. It yields approximately one order of magnitude improvement in relative accuracy at fixed computational cost in this intermediate budget regime.

The remainder of this paper is organized as follows. In Section~\ref{sec:probing} we recall the probing framework and the definition of distance-$d$ colorings in the lattice context. Section~\ref{sec:coloring} gives an overview of both the hierarchical and the multiplier-based coloring schemes and compares the number of colors that the two schemes require for physically relevant lattice geometries. In this section we also investigate situations in which one can prove that
%the assumption that the multiplier-based scheme achieves a strictly larger coloring distance than the completed hierarchical level for the same budget and that the exponential off-diagonal decay of $D^{-1}$ renders the contribution of distant pairs negligible, 
the multiplier-based estimator indeed achieves a lower variance than the partial hierarchical estimator at any intermediate color budget (Lemma~\ref{lem:multiplier_wins}). In Section~\ref{sec:results} we present numerical experiments on a realistic LQCD configuration, comparing the variance reduction achieved by plain Hutchinson, hierarchical probing and multiplier-based probing at fixed computational budget. Section~\ref{sec:conclusions} collects conclusions and an outlook on future work.
\section{Stochastic probing for trace estimation} \label{sec:probing}
Probing, in its original deterministic formulation~\cite{TangSaad2012}, is a technique for extracting the diagonal entries of a sparse matrix by evaluating the quadratic form $v^T A v$ with vectors $v$ constructed from a graph coloring of the matrix sparsity pattern.
Each probing vector $v$ is supported on a single color class, so that the corresponding matrix-vector product simultaneously reveals the entries of $A$ associated with all nodes in that 
class. Since any two distinct nodes within the same color class are non-adjacent in $G$, their contributions to $v^T A v$ are algebraically decoupled from one another. This deterministic approach requires no stochastic sampling and 
produces exact results, but it does not directly extend to the computation of the trace of a matrix inverse of a sparse matrix. We don't have direct access to the inverse as a matrix and, moreover, the inverse is usually a full matrix so that its structure graph is the fully connected graph. A coloring would then need as many colors as the dimension $n$ of the matrix which is prohibitive for larger $n$.
%for which coloring    for (This is so because the inverse of a sparse matrix is usually a full matrix. So  which requires summing all diagonal entries and would need as many probing vectors as there are nodes in the graph in the worst case, and would furthermore require explicit access to the full inverse, which is computationally prohibitive for the case of $D^{-1}$ due to its large size.

The technique used in this work to estimate  $\mathrm{Tr}[D^{-1}]$ is \emph{stochastic probing}~\cite{Wilcox-probing, Stathopoulos-hp}: rather than coloring the graph of $D^{-1}$, which is dense and computationally inaccessible, we instead perform a distance-$d$ coloring of the sparse graph $G$ associated 
with $D$ itself, and then assign i.i.d.\ Rademacher random variables to the nodes within each color class. The trace $\mathrm{Tr}[D^{-1}]$ is then recovered via the Hutchinson estimator applied independently to each class, without ever requiring explicit access to $D^{-1}$. The variance is reduced by the locality structure that $D^{-1}$ inherits from $D$ through the off-diagonal 
decay~\eqref{eq:offdiag-decay}. The result is an unbiased estimator of $\mathrm{Tr}[D^{-1}]$ whose variance is reduced relative to plain Hutchinson by an amount that grows with the coloring distance $d$, as contributions to the variance from all node pairs at graph distance 
at most $d$ are eliminated.

\subsection{Variance reduction from the decay structure of $D^{-1}$} \label{subsec:decay}
Let $G = (V, E)$ denote the undirected graph associated with $D$, where $V = \{1, \ldots, n\}$ indexes the graph nodes (lattice sites in our case) and $(p,q) \in E$ whenever $D_{pq} \neq 0$. It is well known that, for many classes of sparse matrices, the entries of the inverse decay with the graph distance from the main diagonal \cite{Giusti:2019kff}: for the Wilson-Dirac operator one has
\begin{equation}
\label{eq:offdiag-decay}
    |D^{-1}_{pq}| \sim e^{-m_{\pi} \, d_G(p,q)},
\end{equation}
where $d_G(p,q)$ is the graph distance between nodes $p$ and $q$ and $m_{\pi} = am_{\pi}^{\text{phys}}$. Returning to the variance formula~\eqref{eq:variance}, this 
decay means that the dominant contributions to $\mathbb{V}[x^\dagger D^{-1} x]$ come from off-diagonal entries of $D^{-1}$ at short graph distances, with those at distance $d$ 
contributing a term of order $e^{-m_{\pi} d}$ and therefore becoming negligible for large $d$. This suggests that if one could choose the random vectors $x$ in such a way that all 
pairs of nodes at graph distance at most $d$ contribute zero to the quadratic form $x^\dagger D^{-1} x$, the variance would be reduced to the exponentially small tail of~\eqref{eq:variance}, with the dominant short-range terms eliminated
entirely.

This is precisely the idea behind stochastic probing, which relies on a \emph{distance-$d$ coloring} of $G$, that is a map $\mathrm{col} \colon V \to \{1, \ldots, n_c\}$ satisfying
\begin{equation}
\label{eq:distance-d-coloring}
    \mathrm{col}(p) = \mathrm{col}(q) \implies d_G(p,q) > d
    \qquad \forall\, p \neq q \in V,
\end{equation}
so that any two distinct nodes sharing the same color are at graph distance strictly greater than $d$. This is equivalent to a standard (i.e.\ distance-1) vertex coloring of the $d$-th power graph $G^d$, in which two nodes are adjacent whenever $d_G(p,q) \leq d$. The minimum number of colors required for a standard coloring is the {\em chromatic number $\chi$}, and $\chi(G^d)$ is then the chromatic number of a distance-$d$ coloring.
Given such a coloring with $n_c$ color classes $\mathcal{C}_1, \ldots,
\mathcal{C}_{n_c}$, one constructs for each class $m$ a \emph{stochastic probing vector} $v^{(m)}$ through the following: 
\begin{equation}
    v^{(m)}_p = \begin{cases} s_p, & \text{if color}(p) = m   \\ 0, & \text{otherwise} \end{cases},
    \label{eq:probing_vectors}
\end{equation}
with $s_p \sim U(-1,+1)$. The trace estimate is then formed as
\begin{equation}\label{eq:probing_estimator}
    \mathrm{Tr}[D^{-1}] \approx \sum_{m=1}^{n_c} \frac{1}{N_m} \sum_{j=1}^{N_m}(v_j^{(m)})^\dagger D^{-1}\, v_j^{(m)},
\end{equation}
where an independent Hutchinson estimator with $N_m$ samples is applied to each color class $m$.

The variance reduction achieved by this construction follows directly from the coloring property. The variance of the $m$-th term
in~\eqref{eq:probing_estimator} receives contributions only from pairs of nodes within the same color class, so that the total variance is given by:
%\begin{equation}
%\label{eq:probing_variance}
%    \mathbb{V}\!\left[v^\dagger D^{-1}\, v\right]
%    = \frac{1}{2} \sum_{k=1}^{n_c} \sum_{\substack{p \neq q \\ \mathrm{col}(p) = \mathrm{col}(q) = k}}
%    |D^{-1}_{pq} + D^{-1}_{qp}|^2 = \frac{1}{2} \sum_{\substack{p \neq q \\ d_G(p,q) > d}} |D^{-1}_{pq} + D^{-1}_{qp}|^2 ,
%\end{equation}

\begin{equation}
\label{eq:probing_variance}
     \mathbb{V}\!\left[v^\dagger D^{-1}\, v\right]
    = \frac{1}{2} \sum_{m=1}^{n_c} 
    \left\|\mathrm{offdiag}_m(D^{-1} + (D^{-1})^T)\right\|_F^2
     \leq \frac{1}{2} 
    \left\|\mathrm{offdiag}_{>d}(D^{-1} + (D^{-1})^T)\right\|_F^2,
\end{equation}

\noindent
where $\text{offdiag}_m$ denotes the off-diagonal terms of $D^{-1}$ associated with the same color $m$ in the 
corresponding graph, $\text{offdiag}_{>d}$ refers to the off-diagonal entries with a distance from the main diagonal larger 
than $d$. By the distance-$d$ coloring property~\eqref{eq:distance-d-coloring}, any two distinct nodes sharing the same color are at graph distance strictly larger than $d$, so all pairs at distance $\leq d$ contribute zero to the variance. The inequality holds since pairs at distance $> d$ that are assigned different colors by the 
coloring also contribute zero. The variance is therefore bounded above by the sum over all off-diagonal entries at 
graph distance larger than $d$, regardless of their color assignment. This structural reduction of the variance is fundamentally different in nature from what can be achieved by simply increasing the number of samples $N$ in the plain Hutchinson method: no unstructured random sampling can cancel the short-range contributions algebraically. The practical implication is clear: the larger the coloring distance $d$, the smaller the residual variance, but the more colors $n_c$ are needed. The central challenge is therefore to find coloring schemes that achieve large $d$ with as few colors as possible, which is the question we turn to in Section \ref{sec:coloring}.
\section{Coloring of the lattice}\label{sec:coloring}

As established in Section ~\ref{sec:probing}, the effectiveness of stochastic probing relies on the ability to construct distance-$d$ colorings of the lattice graph $G$ with as few colors as possible. On the $n_D$-dimensional hypercubic lattice, $G$ is a regular graph of degree $2n_D$, since each site is connected to exactly two neighbors along each of the $n_D$ lattice directions, one in the positive and one in the negative direction. This lattice has a highly symmetric toroidal structure, and this regularity can be exploited to build colorings analytically rather than resorting to greedy (or other heuristic) graph coloring algorithms. Two such constructions are discussed here, each with distinct properties in terms of the number of colors required and the flexibility with which the coloring distance can be varied. The first, hierarchical probing, is a well-established construction that has been employed in LQCD calculations for over a decade~\cite{Stathopoulos-hp, Stathopoulos-hp2}. The second, the multiplier-based coloring scheme, is a more recent construction that offers greater flexibility in the choice of coloring distance and, as we shall see, proves more beneficial than hierarchical probing in practically relevant computational budget regimes.

\subsection{Hierarchical probing}\label{subsec:hp}
The first very well established construction, known as \emph{hierarchical probing}, introduced in~\cite{Stathopoulos-hp} and extended in ~\cite{Stathopoulos-hp2}, organizes the colorings into a nested sequence of levels $i = 1, 2, 3 \ldots$. On an $n_D$-dimensional torus the number of colors at hierarchical level $i$ is
\begin{equation}
\label{eq:nc-hp}
    n_c^{(\mathrm{HP})}(i) = 2^{n_D(i-1)+1} = 2\,d_i^{n_D},
\end{equation}
where $d_i = 2^{i-1}$ is the coloring distance at level $i$, which makes explicit that the number of colors grows as the $n_D$-th power of the coloring distance.
The $n_c^{(\mathrm{HP})}$ colors are associated with a set of Hadamard vectors $Z^{(i)} = [z_1^{(i)}, \ldots, z_{n_c^{(i)}}^{(i)}]$, with $z_m^{(i)} \in \{-1,+1\}^n$, given by a linear combination of the probing vectors defined in~\eqref{eq:probing_vectors}, such that the following two properties hold:
\begin{enumerate}
    \item [(HP1)] \textbf{Distance property.} The outer product matrix satisfies
    \begin{equation}\label{eq:hadamard_cancel}
    (Z^{(i)}(Z^{(i)})^\dagger)_{pq} 
    = \sum_{m=1}^{n_c^{(i)}} z_m^{(i)}(p)\,\overline{z_m^{(i)}(q)} = 0
    \qquad \text{whenever } d_G(p,q) \leq 2^{i-1},
\end{equation}
where $z_m^{(i)}(p)$ is the $p-$th entry of $z_m^{(i)}$ and $\overline{\,\cdot\,}$ denotes complex conjugation.
    \item [(HP2)] \textbf{Nesting property.} The column space of $Z^{(i-1)}$ is contained in the column space of $Z^{(i)}$, i.e.,\
    $Z^{(i-1)} = Z^{(i)}(:, 1:n_c^{(i-1)})$ up to an orthogonal transformation, so that the quadratures computed at level $i-1$ can be reused at level $i$.
\end{enumerate}

According to the nesting property, the color classes from the preceding level are refined by 
dividing each class into a number of subclasses that depends on $n_D$, so that the partition at 
level $i$ is a refinement of that at level $i-1$. This is the main computational advantage of the 
hierarchical construction: the linear solves performed at level $i-1$ can be reused when moving 
to level $i$, reducing the incremental cost when increasing the probing distance.
The trace estimator at level $i$ is
\begin{equation}\label{eq:hp_estimator}
    \widehat{\mathrm{Tr}}^{(\mathrm{HP},i)}(D^{-1})
    = \sum_{m=1}^{n_c^{(i)}} \frac{1}{N_m} \sum_{j = 1}^{N_m} (z_{m,j}^{(i)})^\dagger D^{-1}\, z_{m,j}^{(i)},
\end{equation}
and by (HP1) its variance is
%\begin{equation}\label{eq:hp_var}
%    \mathbb{V}^{(\mathrm{HP},i)}
%    = \frac{1}{2}\sum_{p \neq q} |D^{-1}_{pq}+D^{-1}_{qp}|^2
%    \cdot |(Z^{(i)}(Z^{(i)})^\dagger)_{pq}|^2
%    = \frac{1}{2}\sum_{\substack{p \neq q \\ d_G(p,q) > 2^{i-1}}}
 %   |D^{-1}_{pq}+D^{-1}_{qp}|^2 \cdot |(Z^{(i)}(Z^{(i)})^\dagger)_{pq}|^2,
%\end{equation}
\begin{equation}\label{eq:hp_var}
    \mathbb{V}^{(\mathrm{HP},i)}
    \leq \frac{1}{2}
    \left\|\mathrm{offdiag}_{>2^{i-1}}(D^{-1} + (D^{-1})^T)\right\|_F^2,
\end{equation}
The vectors $z_{m,j}$ in~\eqref{eq:hp_estimator} are obtained by modulating the Hadamard vectors $z_m$ with independent Rademacher noise vectors $z_j \in \{-1,+1\}^n$, i.e.\ $z_{m,j} = z_j 
\odot z_m$, where $\odot$ denotes the elementwise Hadamard product. This construction makes the Hadamard vectors stochastic while preserving the algebraic structure of the outer product $Z^{(i)}(Z^{(i)})^\dagger$, so that the cancellation guaranteed by property~\eqref{eq:hadamard_cancel} remains intact and ensures that the estimator remains unbiased \cite{Stathopoulos-hp}.
Due to the exponential growth of \eqref{eq:nc-hp} with the hierarchical level $i$, a practitioner operating under a fixed computational budget may be unable to afford the full set of Hadamard vectors at level $i$. A natural remedy is to exploit the nesting property~(HP2) and use only a subset of $s$ vectors, with $n_c^{(i-1)} \leq s < n_c^{(i)}$.

\begin{definition}[Intermediate hierarchical probing] 
Let $i \geq 1$ be a hierarchical level and let $s \in \mathbb{N}$ be an integer satisfying
\[
    n_c^{(i-1)} = 2^{n_D(i-2)+1} \;<\; s \;<\; 2^{n_D(i-1)+1} = n_c^{(i)}.
\]
The \emph{intermediate hierarchical probing estimator} at level $i$ with $s$ colors is defined as
\begin{equation}
    \widehat{\mathrm{Tr}}^{(\mathrm{HP},i,s)}(D^{-1})
    = \sum_{m=1}^{s} \frac{1}{N_m} \sum_{j = 1}^{N_m} (z_{m,j}^{(i)})^\dagger D^{-1}\, z_{m,j}^{(i)}
\end{equation}

where $Z^{(i)}_s = Z^{(i)}(:,1:s)$ denotes the first $s$ columns of the level-$i$ Hadamard matrix $Z^{(i)}$. By the nesting property~\textnormal{(HP2)}, the $s$ vectors contain all $n_c^{(i-1)}$ vectors of level $i-1$ as a subset, so that all quadratures computed at level $i-1$ are reused. 
Let us define the incremental outer product
\begin{equation}\label{eq:delta}
    \Delta_s = Z^{(i)}_s(Z^{(i)}_s)^\dagger - Z^{(i-1)}(Z^{(i-1)})^\dagger
    = \sum_{m=n_c^{(i-1)}+1}^{s} z_m^{(i)}(z_m^{(i)})^\dagger.
\end{equation}
The variance of the partial estimator is then
%\begin{equation}\label{eq:hp_partial_var}
%    \mathbb{V}^{(\mathrm{HP})}(s)
%    = \frac{1}{2}\sum_{p\neq q} |D^{-1}_{pq}+D^{-1}_{qp}|^2
%    \cdot |(Z^{(i-1)}(Z^{(i-1)})^\dagger + \Delta_s)_{pq}|^2.
%\end{equation}
\begin{equation}\label{eq:hp_partial_var}
\begin{aligned}
    \mathbb{V}^{(\mathrm{HP})}(s)
    &= \frac{1}{2}\left\|\mathrm{offdiag}
    \!\left((D^{-1} + (D^{-1})^T) 
    \circ (Z^{(i-1)}(Z^{(i-1)})^T + \Delta_s)\right)\right\|_F^2 \\
    &= \frac{1}{2}\sum_{p\neq q} |D^{-1}_{pq}+D^{-1}_{qp}|^2
    \cdot |(Z^{(i-1)}(Z^{(i-1)})^\dagger + \Delta_s)_{pq}|^2.
\end{aligned}
\end{equation}
\end{definition}

\noindent
It is important to note that the vectors $z_m^{(i)}$ are \emph{not} canonical stochastic probing vectors: unlike the probing vectors~\eqref{eq:probing_vectors}, which have entries in $\{-1, 0, +1\}$ with zeros outside a given color class, the Hadamard vectors $z_m^{(i)} \in \{-1,+1\}^n$ are dense and have no zero entries. 
The variance cancellation for hierarchical probing therefore does not arise from the structural zeros of individual vectors, but from algebraic cancellation in the collective outer product $Z^{(i)}(Z^{(i)})^\dagger$. Specifically, for a pair $(p,q)$ with $d_G(p,q) \leq 2^{i-1}$, each individual term $z_{m,p}^{(i)}\overline{z_{m,q}^{(i)}} = \pm 1$ is nonzero, but the sum over all $n_c^{(i)}$ vectors cancels to zero (see HP1), due to the specific permutation of Hadamard columns chosen by the hierarchical construction~\cite{Stathopoulos-hp}. This cancellation is a collective algebraic property of the complete set of $n_c^{(i)}$ vectors, and it fails for any strict subset of cardinality greater than $n_c^{(i-1)}$. While the completed level $i-1$ retains its own cancellation guarantee up to distance $2^{i-2}$, any intermediate partial set with $n_c^{(i-1)} < s < n_c^{(i)}$ vectors provides no such guarantee beyond distance $2^{i-2}$. Indeed, the coloring distance achieved by such a partial set remains 
$d_{i-1} = 2^{i-2}$: adding the first vectors of level $i$ beyond the completed level $i-1$ assigns different colors to nodes at distance greater than $2^{i-2}$, but not in a way that constitutes a valid distance-$d$ coloring for any 
$d_{i-1} < d < d_i$. The coloring distance therefore stays at $2^{i-2}$ throughout the entire intermediate regime $n_c^{(i-1)} < s < n_c^{(i)}$, and only jumps to $2^{i-1}$ when the full set of $n_c^{(i)}$ vectors is reached.
This is the key structural difference between the Hadamard-based hierarchical estimator and a canonical probing estimator: for the latter, variance cancellation at distance $d$ is guaranteed entry-by-entry by the zeros of the probing vectors, while for the former it relies on the global algebraic structure of the full Hadamard outer product. As a consequence, as $s$ increases from $n_c^{(i-1)}$ to $n_c^{(i)}$, the variance $\mathbb{V}^{(\mathrm{HP})}(s)$ does not decrease monotonically. 
Adding the first vectors of level $i$ beyond the completed level $i-1$ introduces new nonzero off-diagonal contributions to the outer product $\Delta_s$ without yet achieving any additional cancellation, causing the variance to initially rise above that of the completed level $i-1$. Only as $s$ approaches $n_c^{(i)}$ the collective algebraic cancellation begins to take effect, driving the variance back down. At some intermediate value of $s$, the variance crosses below that of the completed level $i-1$, and it is in this regime that the use of intermediate hierarchical probing becomes meaningful in practice.

\subsection{Multiplier-based coloring}\label{subsec:multiplierbased}

The more flexible \emph{multiplier-based} coloring scheme assigns a color to a lattice site $x$ through:
\begin{equation}\label{eq:multiplier_coloring}
    \mathrm{col}(x) = \left(\sum_{i=1}^{n_D} \sigma_i x_i\right) \bmod n_c + 1,
\end{equation}
where $x = (x_1, \ldots, x_{n_D}) \in \mathbb{Z}^{n_D}$ are the integer coordinates of a lattice site and 
$(\sigma_1, \ldots, \sigma_{n_D}) \in \mathbb{N}$ are integer coefficients that depend on the lattice 
dimensions and on the coloring distance $d$. The additive shift of $1$ ensures that colors are indexed 
from $1$ to $n_c$, consistently with the notation used throughout this work. The parameters $\sigma_i$ in a multiplier-based coloring can be determined a priori for each lattice geometry. Computational methods based on efficient variants of an exhaustive search approach are described in~\cite{Langetal:2026}.

The multiplier-based trace estimator uses the probing vectors
$v^{(m)}$ defined in~\eqref{eq:probing_vectors}, and its variance is
%\begin{equation}\label{eq:multiplier_var}
%    \mathbb{V}^{(\sigma)}(n_c)
%    = \frac{1}{2}\sum_{\substack{p \neq q \\ d_G(p,q) > d(n_c)}}
%    |D^{-1}_{pq}+D^{-1}_{qp}|^2 \cdot w_{pq},
%\end{equation}
\begin{equation}\label{eq:multiplier_var}
    \mathbb{V}^{(\sigma)}(n_c) = 
    \frac{1}{2}\sum_{\substack{p \neq q \\ d_G(p,q) > d(n_c)}}
   |D^{-1}_{pq}+D^{-1}_{qp}|^2 \cdot w_{pq},
\end{equation}
where $w_{pq}$ is defined as:
\begin{equation}
    \label{eq:wpq}
    \begin{cases}
        w_{pq} = 1 \quad \text{if col}(p) = \text{col}(q)\\
        w_{pq} = 0 \quad \text{otherwise}
    \end{cases}.
\end{equation}

\begin{table}[h!]
\centering
\begin{tabular}{ccc}
\hline
$d$ & $n_c^{(\sigma)}$ & $n_c^{(HP)}$\\ \hline
1 & 2 & 2\\
2 & 10 & 32\\
3 & 16 & -\\
4 & 64 & 512\\
5 & 128 & -\\
6 & 320 & -\\
7 & 512 & -\\ \hline
\end{tabular}
\caption{Number of colors $n_c$ required for distance-$d$ colorings using the multiplier-based scheme ($n_c^{\sigma}$) and hierarchical probing ($n_c^{\text{HP}}$) on a $64 \times 32^3$ torus. A dash indicates that the distance is not achievable within the hierarchical
framework at that level.} 
\label{tab:colors}
\end{table}

The multiplier-based coloring is not restricted to power-of-2 coloring distances and requires far fewer colors than hierarchical probing at the same 
coloring distance, as shown in Table~\ref{tab:colors} for a $64 \times 32^3$ lattice with $n_D = 4$. It thus provides access to a 
much richer set of intermediate distances. At distance $d = 3$, for example, the multiplier-based scheme requires $16$ colors, a distance that is not 
achievable at all within the hierarchical framework without jumping to the next level at $d = 4$ with $512$ colors. Section~\ref{subsec:comparison} 
further demonstrates that, for any intermediate budget lying between two successive hierarchical levels, we have a theoretical result motivating that the multiplier-based estimator can be expected to achieve a lower variance than the partial hierarchical estimator (Lemma~\ref{lem:multiplier_wins}). A 
detailed numerical comparison of the application of the two schemes to the estimation of observables in LQCD at fixed computational budget is presented in Section~\ref{sec:results}.

\subsection{Theoretical analysis of variance reduction}
\label{subsec:comparison}

Let $D \in \mathbb{C}^{n \times n}$ be a large and sparse matrix and let the entries $D_{ij}^{-1}$ of its inverse obey eq. \eqref{eq:offdiag-decay}. Let $G = (V, E)$ be the undirected graph associated with $D$ as defined in section \ref{sec:probing} represented by an $n_D$-dimensional toroidal lattice with nearest neighbor coupling, and $d_G(i,j)$ the graph distance between nodes $i$ and $j$ in $G$. Given the two probing strategies described in section \ref{sec:coloring}, we   prove the following result:

\begin{lemma}%[multiplier-based achieves strictly lower variance than intermediate hierarchical probing at equal color budget]
\label{lem:multiplier_wins}
Let $n_D \geq 1$ and let $i \geq 1$ be a hierarchical level. Suppose $s=n_c$ is a number of colors satisfying
\begin{equation}\label{eq:intermediate}
    n_c^{(i-1)} = 2^{n_D(i-2)+1} < n_c < 2^{n_D(i-1)+1} = n_c^{(i)},
\end{equation}
and suppose that the multiplier-based coloring with $n_c$ colors achieves a coloring distance $d = d(n_c)$ satisfying 
\begin{equation}
\label{eq:d_condition}
    d(n_c) \geq 2^{i-2} + 1,
\end{equation}
that is, at least by 1 larger than the distance guaranteed by the completed hierarchical level $i-1$.
Then
\begin{equation}\label{eq:multiplier_wins}
    \mathbb{V}^{(\sigma)}(n_c) \leq \mathbb{V}^{(\mathrm{HP})}(n_c) \quad \textit{if the distance d is large enough},
\end{equation}
with strict inequality whenever there exist pairs $(p,q)$ with $2^{i-2} < d_G(p,q) \leq d(n_c)$ that are assigned the same color by the partial hierarchical coloring.
\end{lemma}

\begin{proof}
We decompose the variance of each estimator according to the graph distance of the contributing pairs. For any probing estimator using $n_c$ colors, the variance takes the form
\begin{equation}\label{eq:var_decomp}
    \mathbb{V}(n_c) = \frac{1}{2}\sum_{\ell=1}^{\mathrm{diam}(G)}
    \sum_{\substack{p \neq q \\ d_G(p,q)=\ell}}
    |D^{-1}_{pq}+D^{-1}_{qp}|^2 \cdot w_{pq},
\end{equation}
where $w_{pq} \geq 0$ is a weight that depends on the coloring or probing construction used:
\begin{itemize}
    \item For the multiplier-based estimator, $w^{(\sigma)}_{pq}$ is defined by \eqref{eq:wpq} and is equal to zero for all $\ell \leq d(n_c)$ by the distance-$d(n_c)$ coloring
    property.
    \item For the partial hierarchical estimator, $w_{pq}^{(\mathrm{HP})} =
    |(Z^{(i-1)}(Z^{(i-1)})^\dagger + \Delta_{n_c})_{pq}|^2$. By property (HP1)
    applied to level $i-1$, we have $(Z^{(i-1)}(Z^{(i-1)})^\dagger)_{pq} = 0$
    for $\ell \leq 2^{i-2}$, so for these pairs
    $w_{pq}^{(\mathrm{HP})} = |(\Delta_{n_c})_{pq}|^2$.
\end{itemize}

We now compare the two variances by splitting the sum~\eqref{eq:var_decomp} into three distance ranges.

\medskip
\noindent\textbf{Range 1: $d_G(p,q) \leq 2^{i-2}$.}
For the multiplier-based estimator, all pairs in this range satisfy $d_G(p,q) \leq 2^{i-2} < d(n_c)$ by assumption~\eqref{eq:d_condition}, so $w_{pq}^{(\sigma)} = 0$.

For the partial hierarchical estimator, the completed level $i-1$ guarantees $(Z^{(i-1)}(Z^{(i-1)})^\dagger)_{pq} = 0$ for all pairs at distance $\leq 2^{i-2}$ by~\eqref{eq:hadamard_cancel}. Moreover, since $n_c > n_c^{(i-1)}$, the partial hierarchical estimator includes the full set of level-$(i-1)$ vectors.
We now claim that the cancellation at distance $\leq 2^{i-2}$ is preserved by the partial incremental contribution $\Delta_{n_c}$. To see this, note that the level-$(i-1)$ vectors already eliminate all pairs at 
distance $\leq 2^{i-2}$, and the additional vectors $z_m^{(i)}$ for $m = n_c^{(i-1)}+1, \ldots, n_c$ belong to the refinement of level $i-1$. By the hierarchical construction~\cite{Stathopoulos-hp}, each such vector 
is supported on a subset of nodes that are at distance $> 2^{i-2}$ from each other within the same color class of level $i-1$. Therefore $(\Delta_{n_c})_{pq} = 0$ for all pairs with $d_G(p,q) \leq 2^{i-2}$, and
\begin{equation}\label{eq:range1_corrected}
    (Z^{(i-1)}(Z^{(i-1)})^\dagger + \Delta_{n_c})_{pq}
    = 0
    \qquad \text{for } d_G(p,q) \leq 2^{i-2}.
\end{equation}
Thus the contribution of Range 1 to both variances is exactly zero:
\begin{equation}\label{eq:range1_zero}
    \frac{1}{2}\sum_{\substack{p\neq q \\ d_G(p,q) \leq 2^{i-2}}}
    |D^{-1}_{pq}+D^{-1}_{qp}|^2 \cdot w_{pq}^{(\sigma)}
    = 0
    = \frac{1}{2}\sum_{\substack{p\neq q \\ d_G(p,q) \leq 2^{i-2}}}
    |D^{-1}_{pq}+D^{-1}_{qp}|^2 \cdot w_{pq}^{(\mathrm{HP})}.
\end{equation}
Range 1 therefore contributes no difference between the two variances.

\medskip
\noindent\textbf{Range 2: $2^{i-2} < d_G(p,q) \leq d(n_c)$.}
For the multiplier-based estimator, $w_{pq}^{(\sigma)} = 0$ for all pairs in this range, since the distance-$d(n_c)$ coloring ensures all pairs at distance $\leq d(n_c)$ are in different color classes. 

For the partial hierarchical estimator, the vectors $z_m^{(i)}$ for $m > n_c^{(i-1)}$ are Hadamard vectors with entries from $\{-1,+1\}$. In particular, since the completed level $i-1$ only guarantees cancellation up to distance $2^{i-2}$ via (HP1), the incremental matrix $\Delta_{n_c}$ has no structural reason to vanish for pairs in this distance range. In general, for pairs $(p,q)$ with $2^{i-2} < d_G(p,q) \leq d(n_c)$, we have $|(\Delta_{n_c})_{pq}|^2 > 0$, and therefore $w_{pq}^{(\mathrm{HP})} = |(\Delta_{n_c})_{pq}|^2 > 0$. The contribution of Range~2 satisfies
\begin{equation}\label{eq:range2}
    \frac{1}{2}\sum_{\substack{p\neq q \\ 2^{i-2} < d_G(p,q) \leq d(n_c)}}
    |D^{-1}_{pq}+D^{-1}_{qp}|^2 \cdot w_{pq}^{(\sigma)}
    = 0
    \leq \frac{1}{2}\sum_{\substack{p\neq q \\ 2^{i-2} < d_G(p,q) \leq d(n_c)}}
    |D^{-1}_{pq}+D^{-1}_{qp}|^2 \cdot w_{pq}^{(\mathrm{HP})},
\end{equation}
with strict inequality whenever $|(\Delta_{n_c})_{pq}| > 0$ for at least one such pair, which holds in general for $n_c > n_c^{(i-1)}$.

\medskip
\noindent\textbf{Range 3: $d_G(p,q) > d(n_c)$.}
In this range both estimators receive contributions from pairs at large graph distance. For the multiplier-based estimator, $w_{pq}^{(\sigma)} \in \{0,1\}$ 
depending on whether same-color pairs exist at these distances. For the hierarchical estimator, $w_{pq}^{(\mathrm{HP})} = |(Z^{(i-1)}(Z^{(i-1)})^\dagger + \Delta_{n_c})_{pq}|^2$. Since both estimators use $n_c$ vectors and neither provides structural cancellation 
beyond distance $d(n_c)$ (for the multiplier-based scheme) or $2^{i-1}$ (for the hierarchical scheme, once complete), the contributions of both estimators in this range are of comparable magnitude. In particular, by the exponential decay~\eqref{eq:offdiag-decay}, the terms $|D^{-1}_{pq}+D^{-1}_{qp}|^2$ for $d_G(p,q) > d(n_c)$ are exponentially small compared to those in Range 2, so this range contributes a negligible correction to the difference $\mathbb{V}^{(\mathrm{HP})}(n_c) - \mathbb{V}^{(\sigma)}(n_c)$.

\medskip
\noindent\textbf{Conclusion.}
Combining the three ranges,
\begin{align}\label{eq:final_comparison}
    \mathbb{V}^{(\mathrm{HP})}(n_c) - \mathbb{V}^{(\sigma)}(n_c)
    &= \underbrace{\frac{1}{2}\sum_{\substack{p\neq q \\
    d_G(p,q) \leq 2^{i-2}}}
    |D^{-1}_{pq}+D^{-1}_{qp}|^2 \cdot (w_{pq}^{(\mathrm{HP})} -
    w_{pq}^{(\sigma)})}_{= 0 \text{ by~\eqref{eq:range1_zero}}}
    \nonumber\\
    &\quad +\underbrace{\frac{1}{2}\sum_{\substack{p\neq q \\
    2^{i-2} < d_G(p,q) \leq d(n_c)}}
    |D^{-1}_{pq}+D^{-1}_{qp}|^2 \cdot w_{pq}^{(\mathrm{HP})}}_{
    \geq\, 0,\ \text{strictly positive in general}}
    \nonumber\\
    &\quad + \underbrace{\frac{1}{2}\sum_{\substack{p\neq q \\
    d_G(p,q) > d(n_c)}}
    |D^{-1}_{pq}+D^{-1}_{qp}|^2 \cdot (w_{pq}^{(\mathrm{HP})} - w_{pq}^{(\sigma)})}_{\simeq 0} \geq 0,
\end{align}
which proves~\eqref{eq:multiplier_wins}. Strict inequality holds because Range 2 contributes a strictly positive term: for pairs $(p,q)$ at graph distance $2^{i-2} < d_G(p,q) \leq d(n_c)$, the multiplier-based estimator achieves exact cancellation while the partial hierarchical estimator does not, leaving unremoved variance contributions that are, by the exponential decay~\eqref{eq:offdiag-decay}, larger than those in Range 3 by a factor of at least $e^{m_\pi(d(n_c) - d_G(p,q))}$. The dominant gain of the multiplier-based scheme over hierarchical probing is therefore precisely the variance associated with graph distances in the interval $(2^{i-2}, d(n_c)]$, which multiplier-based cancels exactly and hierarchical probing leaves uncontrolled.
\end{proof}
We note that, as discussed at the end of Section~\ref{subsec:hp}, for values of $n_c$ sufficiently close to $n_c^{(i)}$, where the collective algebraic cancellation of the full Hadamard set begins to take effect, $\mathbb{V}^{(\mathrm{HP})}(n_c)$ may approach $\mathbb{V}^{(\mathrm{HP})}(n_c^{(i)})$ from above, so that the gap between $\mathbb{V}^{(\sigma)}(n_c)$ and $\mathbb{V}^{(\mathrm{HP})}(n_c)$ guaranteed by~\eqref{eq:multiplier_wins} may become small. The bound~\eqref{eq:multiplier_wins} remains valid throughout the entire intermediate range $n_c^{(i-1)} < n_c < n_c^{(i)}$, but its practical significance is largest for values of $n_c$ away from $n_c^{(i)}$, i.e.\ on the increasing side of the hierarchical variance trend described in Section~\ref{subsec:hp}.
A numerical verification of Lemma~\ref{lem:multiplier_wins} on a realistic LQCD configuration is presented in Section~\ref{sec:results}, where the variance reduction achieved by the multiplier-based and hierarchical probing estimators is compared at fixed computational budget across the full range of intermediate colorings.
\subsection{Dilution}\label{subsec:dilution}

So far we have treated the graph nodes as scalar degrees of freedom, assigning a single $\pm 1$ entry to each node in the construction of the probing vectors~\eqref{eq:probing_vectors}. In LQCD, however, our graph is a torus where lattice site $x$ carries $N_s N_c = 12$ internal degrees of freedom, corresponding to $N_s = 4$ spin components and $N_c = 3$ color charges, 
as described in Section~\ref{sec:intro}. As a consequence, the probing vectors $v^{(m)} \in \mathbb{C}^n$ live in a space of dimension $n = 12|V|$, and the $12$ internal degrees of freedom at each site are in general all simultaneously
nonzero in a probing vector constructed from a coloring of the site graph alone.

The coupling between internal degrees of freedom introduces additional off-diagonal contributions to the variance~\eqref{eq:probing_variance} that the graph coloring does not address. These contributions arise because the internal degrees of freedom at a single site are not separated by any graph distance, so the distance-$d$ coloring alone cannot cancel their contribution to the variance. 

\emph{Dilution}~\cite{Foley-dilution} addresses this by further splitting each probing vector $v^{(m)}$ into a set of vectors that have support on only a prescribed subset of the internal degrees of freedom at each site. Formally, let $\mathcal{P} = \{P_1, \ldots, P_{n_d}\}$ be a partition of the index set $\{1, \ldots, 12\}$ of the internal degrees of freedom into $n_d$ disjoint subsets.
The diluted probing vectors associated with color class $m$ and dilution subset $\ell$ are defined by
\begin{equation}\label{eq:dilution_vectors}
    v^{(m,\ell)}_p =
    \begin{cases}
        s_p & \text{if } \mathrm{col}(\lfloor p/12 \rfloor) = m
               \text{ and } (p \bmod 12)+1 \in P_\ell, \\
        0   & \text{otherwise,}
    \end{cases}
\end{equation}
where $s_p \sim U(-1,1)$ i.i.d., and the trace estimator becomes
\begin{equation}\label{eq:diluted_estimator}
    \mathrm{Tr}[A] \approx \sum_{m=1}^{n_c} \sum_{\ell=1}^{n_d} \frac{1}{N_{m}} \sum_{j=1}^{N_{m}}
    (v_j^{(m,\ell)})^\dagger D^{-1}\, v_j^{(m,\ell)}.
\end{equation}
%increasing the number of linear system solves by a factor of $n_d$.
%Since the estimator~\eqref{eq:diluted_estimator} now contains $n_d \cdot n_c$ terms, each requiring one linear solve, a fixed computational budget that previously afforded $N$ samples per color class now affords $N/n_d$ samples per color-dilution pair; the increased number of terms however targets additional sources of variance, so that the overall variance is reduced despite the smaller sample count per term.

The choice of partition $\mathcal{P}$ defines the \emph{dilution scheme}. Three natural schemes are in common use. In \emph{spin dilution}, the $N_s = 4$ spin components are separated while the $N_c = 3$ color components---these are inner quark degrees of freedom not to be confused with the colors used in probing---are left unseparated. This results in  $n_d = N_s = 4$ dilution vectors per color class. In \emph{color dilution},
the roles are reversed: the $N_c = 3$ color components are separated while the spin components are left together, giving $n_d = N_c = 3$. In \emph{full dilution}, all $N_s N_c = 12$ internal degrees of freedom are treated independently, so that each diluted vector has support on exactly one spin-color component at each lattice site of the corresponding color class, yielding $n_d = 12$. Full dilution maximally separates the internal degrees of freedom and therefore achieves the largest variance reduction among the three schemes, 
at the cost of introducing $n_d = 12$ terms per color class in ~\eqref{eq:diluted_estimator}; for a fixed computational budget this means fewer stochastic samples $N_m$ per term, but the reduction in variance per term typically compensates for this.
Partial dilution schemes such as spin or color dilution offer intermediate trade-offs between variance reduction and computational cost.

The interplay between stochastic probing and dilution is additive in the sense that they target distinct sources of variance: stochastic probing eliminates the off-diagonal contributions from pairs of sites at short graph distance, while dilution eliminates the contributions from the coupling between internal degrees of freedom at the same or nearby sites. The two techniques are therefore complementary and can be combined
freely. An optimal choice of coloring distance $d$ and dilution scheme depends on the relative magnitudes of the spatial and internal contributions to the variance, which in turn depend on the conditioning of $D$ and the lattice geometry.

\section{Numerical Results} \label{sec:results}

We now compare the two coloring strategies, i.e.\ hierarchical probing and multiplier-based probing, to plain Hutchinson for the estimation of $\mathrm{Tr}[D^{-1}]$ on a realistic LQCD gauge configuration. 
We measure the relative accuracy $\epsilon^2$ as a function of the number of colors $n_c$ at a fixed computational budget, measured as the total number of linear system solves, on a $64 \times 32^3$ lattice with $m_{\pi} = am_{\pi}^{\text{phys}} \simeq 0.3322$.
The Wilson-Dirac operator $D$ thus has matrix dimension $n = 64 \times 32^3 \times 12 \approx 25 \times 10^6$, and we apply the full dilution scheme as described in Section~\ref{subsec:dilution}. The linear solves for the samples~\eqref{eq:diluted_estimator} 
were obtained using the DD$\alpha$-AMG multigrid solver~\cite{DDalphaAMG-paper}. The accuracy $\epsilon^2$ is computed as the empirical variance of the estimator \eqref{eq:diluted_estimator} given as
\begin{equation}
    \epsilon^2 = \sum_{m=1}^{n_c} \frac{1}{N_m} \sum_{l=1}^{12} \widehat{\mathbb{V}}_{m,l},
\end{equation}
where $\widehat{\mathbb{V}}_{m,l}$ represents the sample variance associated with the $l$-th degree of freedom for the $m$-th color component, estimated from the $N_m$ stochastic samples as
\begin{equation}
    \widehat{\mathbb{V}}_{m,l} = \frac{1}{N_m-1} \sum_{k=1}^{N_m} 
    \left( e_{m,l}^{(k)} - \bar{e}_{m,l} \right)^2, 
    \qquad \bar{e}_m = \frac{1}{N_m} \sum_{k=1}^{N_m} e_{m,l}^{(k)},
\end{equation}
where $e_{m,l}^{(k)} = (z_{m,l}^{(k)})^\dagger D^{-1} z_{m,l}^{(k)}$ denotes the $k$-th sample of the estimator corresponding to the $l$-th degree of freedom of the $m$-th color. The quantity $\epsilon^2$ is therefore itself a random variable; the values reported in Figure~\ref{fig:trace} are single-run estimates 
obtained from a fixed set of samples and should be interpreted as empirical approximations to the ``true'' variance of the estimator 
 $\sum_{m=1}^{n_c} \frac{1}{N_m} \sum_{l=1}^{12} \mathbb{V}_{m,l}$, where $\mathbb{V}_{m,l}
=
\mathbb{E}\!\left[
\left|
e_{m,l}^{(k)} - \mathbb{E}[e_{m,l}^{(k)}]
\right|^2
\right]$ is the theoretical variance of the $m$-th color estimator. For all methods and all values of $n_c$ shown in Figure~\ref{fig:trace}, 
the number of samples $N_m$ per color class $m$ is chosen to be uniform across all spin and color indices within the same class, and the total number of linear system solves $N_{\mathrm{tot}} = 12\sum_{m=1}^{n_c} N_m$ is held fixed across all data points for a chosen value of $N_{\text{tot}}$, ensuring a comparison at equal computational cost.

The first observation is that both stochastic probing schemes outperform plain Hutchinson substantially across the entire range of $n_c$, while $\epsilon^2$ gradually decreases as $n_c$ grows and more off-diagonal contributions are cancelled. This confirms the effectiveness of stochastic probing as a variance reduction strategy: by constructing sampling vectors that are orthogonal on all node pairs at graph distance at most $d$, both methods eliminate the contributions to the variance arising from 
off-diagonal entries of $D^{-1}$ at small graph distance, which dominate due to the exponential decay~\eqref{eq:offdiag-decay} and which the plain Hutchinson estimator is unable to suppress.

\begin{figure}[h!]
    \centering
    \includegraphics[width=0.6\textwidth]{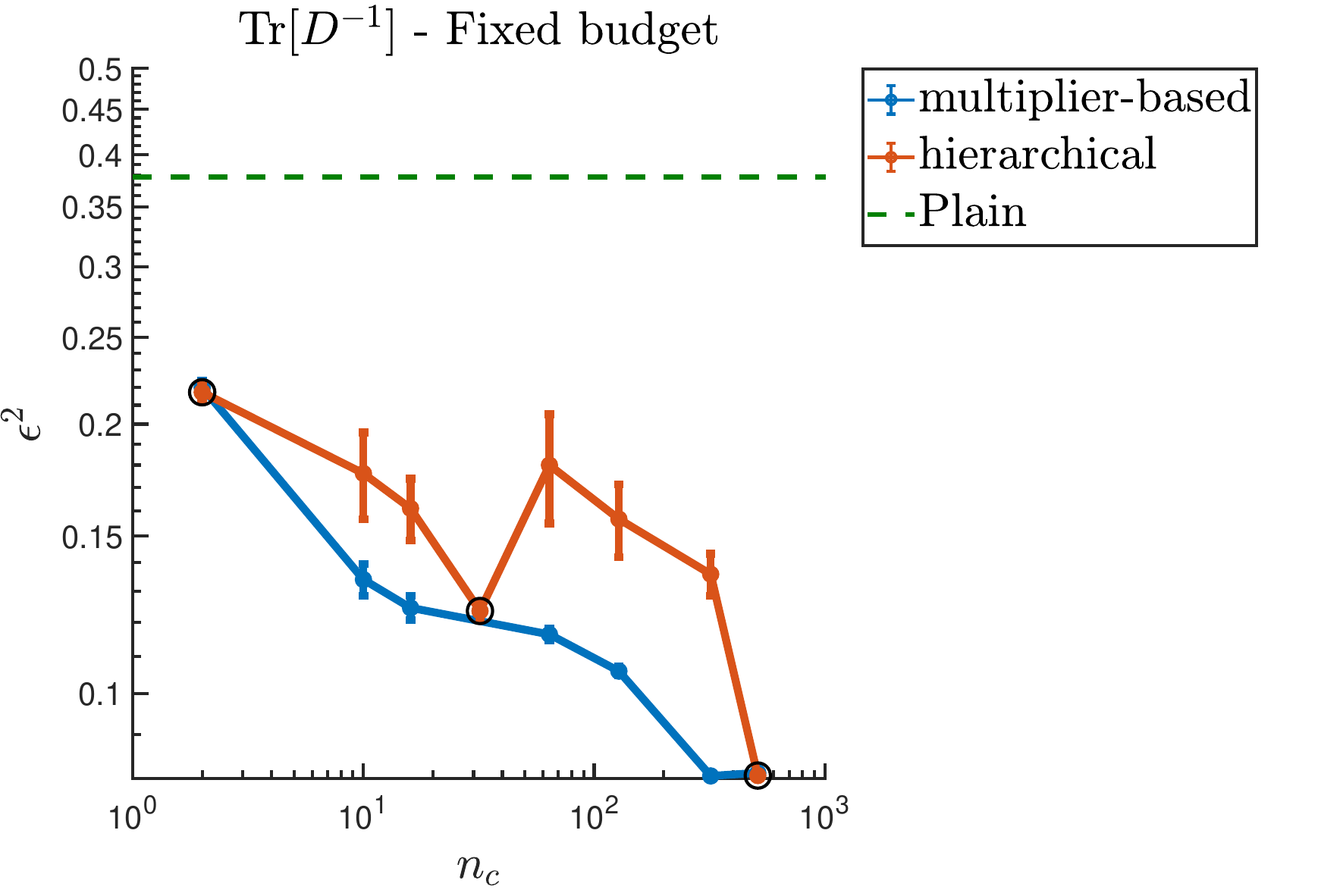}
    \caption{Relative accuracy $\epsilon^2$ as a function of the number of colors $n_c$ for the estimation of $\mathrm{Tr}[D^{-1}]$ on a $64 \times 32^3$ torus at fixed computational budget. Multiplier-based probing (blue), hierarchical probing (orange), and plain Hutchinson (dashed green) are compared. Circled markers on the hierarchical probing curve indicate the power-of-two coloring distances for which a comparable accuracy between the two probing methods is achieved.}
    \label{fig:trace}
\end{figure}

We now turn to the comparison between the two probing schemes. At the values $n_c =2$ and $512$ corresponding to a hierarchical probing power-of-two coloring 
distances $d = 1$ and $d=4$, the two methods achieve essentially the same variance reduction 
using the same number of colors. The coloring distances achieved in the multiplier-based approach are
significantly larger than those in hierarchical probing; see Table~\ref{tab:colors}, and this is consistent with assumption~\eqref{eq:d_condition} of Lemma~\ref{lem:multiplier_wins}. 
These coinciding points are highlighted by circled markers in Figure~\ref{fig:trace}. 

The middle circled point is an exception: it corresponds to a distance-$2$ coloring under hierarchical probing which uses 
$32$ colors, a number of colors that falls between those required by the multiplier-based distance-$3$ and distance-$4$ colorings on a $64 \times 32^3$ 
torus (cf.\ Table~\ref{tab:colors}). We thus  cannot directly compare with a multiplier-based coloring using the same number of colors. We see, however, that in this situation the variance achieved with hierarchical probing is quite close to that for multiplier-based coloring for a distance-3 coloring which requires 16 colors only. 

% however,  of equal cost achieves the same distance. We note, however, that while the achieved relative accuracy coincides at these points, the two methods do so with the same number of colors but different coloring distances: 
% the multiplier-based scheme achieves a strictly greater coloring distance than hierarchical probing for the same $n_c$, consistently with assumption~\eqref{eq:d_condition} of Lemma~\ref{lem:multiplier_wins}. 
The agreement in accuracy confirms that, as made precise by equation~\eqref{eq:probing_variance}, what determines the variance is not the coloring distance per se but the set of node pairs whose contributions to the variance are eliminated. The additional pairs eliminated by the larger multiplier-based distance contribute entries of $D^{-1}$ that are exponentially suppressed by~\eqref{eq:offdiag-decay} and therefore negligible in practice.

The difference between the two methods emerges at the intermediate values of $n_c$, and it is here that the numerical results confirm the theoretical prediction of Lemma~\ref{lem:multiplier_wins}. In many realistic LQCD calculations the available computational budget 
typically allows for $\mathcal{O}(10^2)$-$\mathcal{O}(10^3)$ linear solves per observable~\cite{Gambhir2016, Green2017}, 
typically leading to an even lower per-trace budget once the cost is distributed across the multiple operators contributing to the 
observable. In the hierarchical probing framework, this computational budget often falls short of the minimum $n_c^{(\mathrm{HP})} = 2^{n_D(i-1)+1}$ solves required to complete the next hierarchical level, and one is forced to work with a number of vectors that lies between two consecutive power-of-two checkpoints. In many LQCD calculations the budget is so tight that only a single noise sample per color ($N = 1$) is affordable, so that variance control relies entirely on the coloring distance achieved by the probing scheme rather than on stochastic averaging. This makes the choice of coloring scheme 
particularly consequential: a scheme that guarantees a larger coloring distance for the same number of vectors directly translates into a lower theoretical variance $\mathbb{V}_m$, even if this variance remains unobservable within a single sample. In this context, as established in Section~\ref{subsec:comparison}, using a partial set of level-$i$ vectors does not constitute a valid distance-$2^{i-1}$ coloring, and the variance 
reduction is correspondingly degraded relative to what the completed hierarchical level would achieve.

The multiplier-based method is entirely free from this trade-off, at the cost of not supporting the reuse of solutions across different values of $n_c$, since it does not possess a nesting property analogous to~(HP2). Since every value of $n_c$ in the multiplier-based scheme corresponds to a genuine distance-$d$ coloring, the relative accuracy decreases smoothly and monotonically with $n_c$, with no intermediate loss of quality, precisely as predicted by the strict variance inequality~\eqref{eq:multiplier_wins}. At intermediate values of $n_c$ this translates into a substantial improvement in relative accuracy over hierarchical probing anticipated in Section~\ref{sec:coloring}, precisely in the budget regime that is most commonly encountered in LQCD calculations.
\section{Conclusions} \label{sec:conclusions}

In this work we have addressed the problem of estimating $\mathrm{Tr}[D^{-1}]$, where $D$ is the Wilson-Dirac matrix of LQCD.
Since the direct computation of $D^{-1}$ is entirely infeasible for the matrix sizes encountered in realistic LQCD simulations, we have employed the Hutchinson stochastic estimator combined with stochastic probing. The latter replaces unstructured random vectors with structured sampling vectors constructed from distance-$d$ colorings of the lattice graph, exploiting the exponential off-diagonal decay of $D^{-1}$ to cancel the dominant contributions to the variance from nearby node pairs algebraically.

The central numerical finding of this work is that the recently developed multiplier-based coloring scheme provides a more flexible and, at any intermediate computational budget, more accurate alternative to hierarchical probing. The two methods agree in accuracy when the number of colors corresponds to 
a hierarchical power-of-two coloring distance, where both produce a valid distance-$d$ coloring and therefore achieve identical variance reduction, even though the multiplier-based scheme achieves a strictly larger coloring distance for the same number of colors. At all intermediate number of colors, hierarchical probing is forced to rely on partial colorings that do not achieve a valid graph distance, resulting in an irregular and non-monotonic behavior of the accuracy as a function of the number of colors. The multiplier-based scheme, by contrast, guarantees a genuine distance-$d$ coloring at every value of $n_c$. This yields a smooth monotonic reduction of the variance and a visible and consistent improvement in relative accuracy over hierarchical 
probing at fixed computational budget in the intermediate regime. Both methods represent a substantial improvement over plain Hutchinson estimation, confirming the effectiveness of stochastic probing as a variance reduction strategy for trace estimation on sparse matrices with off-diagonal decay.

Looking ahead, the methods applied here open a natural path towards the estimation of other physically relevant quantities in LQCD. The most immediate target is the computation of $\mathrm{Tr}[\Gamma_5 D^{-1}]$ and its time-slice restriction $\mathrm{Tr}[\Gamma_5(t) D^{-1}(t,t)]$: since 
$\Gamma_5$ acts locally on the spin degrees of freedom at each lattice site, the 
graph structure underlying the probing construction is unchanged, and the multiplier-based coloring scheme applies without modification. These operators are directly related to the pseudoscalar condensate and to the disconnected contribution to the quark loop, key quantities in the study of chiral symmetry breaking and the QCD phase structure. Additionally, an important further direction is the extension of the multiplier-based coloring scheme to twisted mass fermion formulations, where the Wilson-Dirac operator is augmented by a chirally twisted mass term and the structure of the relevant traces is modified accordingly.

\begin{acknowledgement}
This work is supported by AQTIVATE (\url{https://aqtivate.ucy.ac.cy/}) which has received funding from the European Union’s research and innovation programme under the Marie Skłodowska-Curie Doctoral Networks action and Grant Agreement No 101072344. The computations for the measurements of the variance were carried out on the PLEIADES cluster at the University of Wuppertal, which was supported by the Deutsche Forschungsgemeinschaft (DFG, grant No. INST 218/78-1 FUGG) and the Bundesministerium für Bildung und Forschung (BMBF). G.R-H. acknowledges financial support from the EoCoEIII project, which has received funding from the European High Performance Computing Joint Undertaking under grant agreement No. 101144014. The authors gratefully acknowledge the Gauss Centre for Supercomputing e.V. (\url{www.gauss-centre.eu}) for funding this project by providing computing time through the John von Neumann Institute for Computing (NIC) on the GCS JUWELS at Jülich Supercomputing  Centre (JSC).
\end{acknowledgement}

\vspace{\baselineskip}
%% The style of the following references should be used in all documents.

\bibliographystyle{unsrt}
\bibliography{biblio}

%\begin{thebibliography}{1}

%\bibitem{bib1}% 
 %F.\,M. Firstauthorfamilyname, F.\,M. Secondauthorfamilyname, and
  %C.~Lastauthorfamilyname,
 %Abbreviatedjournalname \textbf{volume}, page (year).

%\bibitem{bib2}% 
 %F.~Examplename and  I.\,E. Anotherauthorname,
 %phys. stat. sol. (a) \textbf{1}, 111 (2050).
%\end{thebibliography}

\end{document}